\documentclass[12pt]{article}
\RequirePackage[colorlinks,citecolor=blue,urlcolor=blue,linkcolor=blue]{hyperref}
\usepackage{url}

\usepackage[ansinew]{inputenc}

\usepackage{doi}

\usepackage{amsmath}

\usepackage{mathtools}

\usepackage{amsfonts}
\usepackage{amsthm}
\usepackage{amssymb}

\usepackage{color}
\usepackage{bm} 
\usepackage{dsfont} 

\usepackage{graphicx}
\usepackage[lmargin=2.5cm,rmargin=2.5cm,bottom=3cm,top=3cm]{geometry}

\usepackage{tikz}
\definecolor{mycolor1}{rgb}{0.105882,0.619608,0.466667}
\definecolor{mycolor2}{rgb}{0.85098,0.372549,0.00784314}
\definecolor{mycolor3}{rgb}{0.458824,0.439216,0.701961}
\definecolor{mycolor4}{rgb}{0.905882,0.160784,0.541176}
\definecolor{mycolor5}{rgb}{0.4,0.65098,0.117647}
\definecolor{mycolor6}{rgb}{0.65098,0.462745,0.113725}
\definecolor{mycolor7}{rgb}{0.901961,0.670588,0.00784314}
\definecolor{mycolor8}{rgb}{0.4,0.4,0.4}
\definecolor{mycolor9}{rgb}{0.301961,0,0.294118}
\definecolor{mycolor10}{rgb}{0.0313725,0.25098,0.505882}

\usetikzlibrary{calc}        

\makeatletter
\newif\ifmygrid@coordinates
\tikzset{/mygrid/step line/.style={line width=0.80pt,draw=gray!80},
         /mygrid/steplet line/.style={line width=0.25pt,draw=gray!80}}
\pgfkeys{/mygrid/.cd,
         step/.store in=\mygrid@step,
         steplet/.store in=\mygrid@steplet,
         coordinates/.is if=mygrid@coordinates}
\def\mygrid@def@coordinates(#1,#2)(#3,#4){%
    \def\mygrid@xlo{#1}%
    \def\mygrid@xhi{#3}%
    \def\mygrid@ylo{#2}%
    \def\mygrid@yhi{#4}%
}
\newcommand\DrawGrid[3][]{%
    \pgfkeys{/mygrid/.cd,coordinates=true,step=1,steplet=0.2,#1}%
    \draw[/mygrid/steplet line] #2 grid[step=\mygrid@steplet] #3;
    \draw[/mygrid/step line] #2 grid[step=\mygrid@step] #3;
    \mygrid@def@coordinates#2#3%
    \ifmygrid@coordinates%
        \draw[/mygrid/step line]
        \foreach \xpos in {\mygrid@xlo,...,\mygrid@xhi} {%
          (\xpos,\mygrid@ylo) -- ++(0,-3pt)
                              node[anchor=north] {$\xpos$}
        }
        \foreach \ypos in {\mygrid@ylo,...,\mygrid@yhi} {%
          (\mygrid@xlo,\ypos) -- ++(-3pt,0)
                              node[anchor=east] {$\ypos$}
        };
    \fi%
}
\makeatother

\usepackage{psfrag}

\usepackage{algorithm}
\usepackage{algpseudocode}
\algdef{SE}[DOWHILE]{Do}{doWhile}{\algorithmicdo}[1]{\algorithmicwhile\ #1}%

\usepackage{mathrsfs} 

\usepackage[square,numbers,compress]{natbib}

\newcommand{\remove}[1]{}
\newcommand{\removesafe}[1]{}

\usepackage{subfigure}

\newcommand{\homegac}{h_{\Omega^c}}
\newcommand{\homega}{h_{\Omega}}
\DeclareMathOperator*{\sign}{\text{sgn}}
\newcommand{\xo}{x_0}

\newcommand{\hxop}{h_{x_0^\perp}}
\newcommand{\R}{\mathbb{R}}

\newcommand{\xoomega}{x_0}
\newcommand{\xhat}{\hat{x}}
\newcommand{\ctilde}{\tilde{c}}
\newcommand{\Ctilde}{\tilde{C}}

\newtheorem{theorem}{Theorem} 

\newtheorem{corollary}[theorem]{Corollary}

\title{Compressed Sensing from Phaseless Gaussian Measurements via Linear Programming in the Natural Parameter Space }
\author{
Paul Hand\footnote{Department of Computational and Applied Mathematics, Rice University,  Houston, TX.} \   and Vladislav Voroninski\footnote{Helm.ai, CA.}}

\begin{document}

\maketitle
\abstract{We consider faithfully combining phase retrieval with classical compressed sensing. Inspired by the recent novel formulation for phase retrieval called PhaseMax, we present and analyze SparsePhaseMax, a linear program for phaseless compressed sensing in the natural parameter space. We establish that when provided with an initialization that correlates with an arbitrary $k$-sparse $n$-vector, SparsePhaseMax recovers this vector up to global sign with high probability from $O(k \log \frac{n}{k})$ magnitude measurements against i.i.d. Gaussian random vectors. Our proof of this fact exploits a curious newfound connection between phaseless and 1-bit compressed sensing. This is the first result to establish bootstrapped compressed sensing from phaseless Gaussian measurements under optimal sample complexity.}

\section{Introduction}

Since the foundational compressed sensing results of Candes, Tao, and Donoho over 10 years ago, it is now common in the physical sciences to exploit the sparse structure of natural signals to enable signal recovery from fewer measurements \cite{CRT2005, donoho2001uncertainty}. While sample-optimal compressed sensing from linear measurements is well-developed both mathematically and empirically, enabling compressed sensing in non-linear measurement scenarios has proven to be a substantial challenge. In particular, we consider here phaseless compressed sensing, wherein a sparse vector is to be recovered from magnitude observations only.    In recent years, the phase retrieval problem has garnered significant attention in the applied mathematics community since the development of the PhaseLift algorithm \cite{CESV2011, CSV2013}, which is a convex program that operates in a lifted space of matrices. The PhaseLift methodology has also been applied to recovering sparse vectors from phaseless measurements but was proven to necessarily operate at suboptimal sample complexity in this regime  \cite{li2013sparse}. Due to the computational intractability of working in lifted spaces and the suboptimal sample complexity of PhaseLift in the sparse setting, much work has focused on alternative formulations to phase retrieval, including a recent trend of several non-convex approaches \cite{altminphase, wirtinger, twf, SQW2016, yoninaTAF}.

Concretely, compressed sensing is as follows.  Let $x_0 \in \mathbb{R}^n$ be $k$-sparse and consider a set of known measurement vectors $a_i \in \mathbb{R}^n, i = 1\ldots m$. We consider here the standard idealized setting where $a_i$ are i.i.d Gaussian. It was shown in \cite{CRT2005} and \cite{Donoho}, that $x_0$ can be recovered with high probability via linear programming from $m = O(k \log n)$ linear measurements $\langle a_i, x_0 \rangle, i=1\ldots m$, yielding an order-optimal sample complexity in $k$. 

In this paper, we consider performing compressed sensing and phase retrieval simultaneously. That is, we aim to recover $x_0$ up to a global sign in polynomial time from the magnitude-only observations $|\langle a_i, x_0 \rangle |, i =1\ldots m$.  Naturally, we seek to minimize the number of observations required.  Given an initialization vector $\phi \in \R^n$ that correlates with the signal $x_0$, we consider the following linear program, called SparsePhaseMax:

\begin{equation}\label{PhaseMax}
\begin{array}{ll}
\max & \phantom{-|}\langle \phi, x \rangle - \lambda \|x\|_1 \\[.5em]
\text{s.t.} & -|\langle a_i, x_0 \rangle | \leq \langle a_i, x \rangle \leq |\langle a_i, x_0 \rangle |, \quad  i = 1\ldots  m, \quad x \in \mathbb{R}^n.\\ 

\end{array}
\end{equation}

We prove that if the  initialization vector (also known as an anchor vector) $\phi$  correlates with $\xo$, then SparsePhaseMax with an appropriate choice of $\lambda$ recovers $\xo$ exactly with high probability from $m = O(k \log n/k)$ phaseless Gaussian measurements, provided that $x_0$ is sufficiently sparse. 
We do not make any structural assumptions on $\xo$ beyond sparsity.  In a prior version of the present paper, we proved our main result under the assumption that the magnitude of the nonzero entries of $\xo$ were constant.  For the theorem and proof below, a simple modification removed this assumption. Our main results are as follows:

\begin{theorem} \label{thm:sparsephasemax-smallalpha}
Let $\xo \in \R^n$ be a $k$-sparse vector. Let $a_i \sim \mathcal{N}(0, I_{n \times n})$ be independent for $i = 1 \ldots m$ and assume that $|\langle a_i, x_0 \rangle |, i = 1\ldots m$ have been observed. Let $\phi$ be a unit vector such that $\left \langle \phi, \frac{\xo}{\| \xo \|_2} \right \rangle =: \alpha >0$.  Let $\lambda$  be such that $\frac{1}{2} \frac{\alpha}{\sqrt{k}} < \lambda < \frac{3}{4}  \frac{\alpha}{\sqrt{k}}$, and take ${k < \frac{1}{49} \alpha^2 n}$. Then, provided that $m > \frac{C_1}{\alpha^7} k \log \frac{n}{k}$, $\xo$ is the unique maximizer of SparsePhaseMax \eqref{PhaseMax} with probability at least $1 - C_2 e^{-c \alpha^{-2/5} m^{4/5} k^{1/5} }$.  Here, $C_1, C_2, c$ are positive universal constants. 
\end{theorem}

If the correlation is bounded below by a fixed constant, we have the following.

\begin{corollary} \label{cor:sparsephasemax-smallalpha}
Fix $\gamma\in \R$ such that $0 < \gamma \leq 1$.  Let $\xo \in \R^n$ be a $k$-sparse vector. Let $a_i \sim \mathcal{N}(0, I_{n \times n})$ be independent for $i = 1 \ldots m$. and assume that $|\langle a_i, x_0 \rangle |, i = 1 \ldots m$ have been observed.  Let $\phi$ be a unit vector such that $\left \langle \phi, \frac{\xo}{\| \xo \|_2} \right \rangle=: \alpha > \gamma$, and take $k < \frac{1}{49}\alpha^2 n$.  Let $\lambda$ be such that $\frac{1}{2} \frac{\alpha}{\sqrt{k}} < \lambda < \frac{3}{4}  \frac{\alpha}{\sqrt{k}}$. Then $\xo$ is the unique maximizer of SparsePhaseMax \eqref{PhaseMax} with probability at least $1 - C_2 e^{-c  k}$, for $m = O(k \log \frac{n}{k})$.  Here, $C_1, C_2, c$ are positive universal constants. 
\end{corollary}

Thus, pre-initialized phaseless compressed sensing with random Gaussian measurement vectors is possible under optimal sample complexity.

\section{Comparison to prior work}

There has been much previous work on the phaseless compressed sensing problem.  See \cite{cpremprical} for an empirical exploration.  In terms of theory, \cite{li2013sparse} established that $x_0$ can be recovered modulo sign via an L1-enhanced PhaseLift formulation when $m = O(k^2 \log n)$. Simultaneously, \cite{li2013sparse} established that this result is sharp. Thus, there is a substantial gap between optimal sample complexity achievable with PhaseLift in the sparse setting and the information theoretic lower bound of $O(k)$ measurements required for injectivity over $k$-sparse vectors of Gaussian magnitude measurements \cite{voroninski2016strong}.

Much subsequent work on the sparse phase retrieval problem has followed but has been unable to break through the algorithmic barrier described above \cite{oymak2015simultaneously, GAMP, Yon}. There are also related papers which consider measurement vectors that have a specific algebraic structure \cite{DustinSparse, iwen2015robust, jaganathan2013sparse} or very sparse measurement vectors \cite{Babak}. In these settings, it is possible to achieve optimal sample complexity. In fact, if we allow ourselves arbitrary choice of sensing vectors, optimal sample complexity follows simply by choosing the measurement matrix as a factorized product of Gaussian matrices of appropriate dimensions and applying PhaseLift followed by standard basis pursuit. The issue with these approaches is that it is not realistic to expect such flexibility in measurement vector design in applied science scenarios, and moreover the assumed algebraic measurement structure avoids the central difficulty of sparse phase retrieval.  Meanwhile, the generic Gaussian model is much closer to modeling the physical reality of Fourier diffraction measurements in X-ray crystallography and other areas of diffraction imaging. That is, once the Gaussian setting is understood, extensions to the Fourier setting are very likely to follow via similar methods. 

A paper by Li et al. on Truncated Wirtinger Flow (TWF) \cite{cai2015optimal} examines the case of unstructured Gaussian measurement vectors for sparse phase retrieval in the noisy setting. The authors of \cite{cai2015optimal} combine initialization using $O(k^2 \log n)$ Gaussian magnitude measurements with a non convex optimization scheme in the natural parameter space, establishing theoretical guarantees for the latter which also require $O(k^2 \log n)$ measurements. Our approach also relies on an initialization, which we utilize as an anchor vector, but differs in several important aspects. Firstly, unlike TWF, which consists of a non-convex objective minimization, we utilize \emph{convex} programming in the natural parameter space. In fact our approach consists of a linear program closely related to the recent PhaseMax program for regular phase retrieval \cite{phasemaxJustin, phasemax, phasemaxelementary}. Secondly, while theoretical guarantees for TWF require $O(k^2 \log n)$ measurements for non-convex optimization after initialization, we show that our approach succeeds post-initializion at the optimal sample complexity of $O(k \log \frac{n}{k})$. Thirdly, our theoretical arguments are far simpler than those employed in non-convex approaches to phase retrieval and our algorithmic formulation has only one parameter, for which a range of values is admissible. To the best of our knowledge, our main result is the first of its kind with respect to any of the three prior points in the Gaussian phaseless compressed sensing setting.


There is some conflicting evidence regarding the optimal sample complexity of phaseless compressed sensing. For instance, while PhaseLift is suboptimal in the sparse setting, the authors of \cite{voroninski2016strong} showed that minimizing the L1 norm over Gaussian magnitude measurements, if it were possible in polynomial time, would achieve sample-optimal phaseless compressed sensing. At the same time, Sparse PCA \cite{berthet2013optimal} and Planted Clique \cite{barak2016nearly}, which may be intimately related to phaseless compressed sensing, have resisted attempts to break past barriers analogous to the $O(k^2 \log n)$ Gaussian phaseless measurement regime. Our work here meanwhile indicates that it is possible to recover in polynomial time a sparse vector from $O(k \log \frac{n}{k})$ Gaussian magnitude measurements, provided that we start with an initialization that sufficiently correlates with the true solution. Subsequent to the release of this work, \cite{wang2016sparse} showed that a Truncated Amplitude Flow can also recover a sparse vector from $O(k)$ phaseless Gaussian measurements, provided an appropriate initialization is known.  Similarly, \cite{soltanolkotabi2017structured} showed similar results for the Projected Wirtinger Flow for signals with a variety of structural assumptions, including sparsity.  We note that the best known initialization scheme for sparse phase retrieval is only known to provide meaningful initializations at the sample complexity of $\Omega(k^2 \log n)$ Gaussian measurements \cite{cai2015optimal}. It is a very interesting open problem to either come up with a sample-optimal initialization scheme for sparse phase retrieval, which by appealing to our main theorem is equivalent to sample-optimal sparse phase retrieval, or to prove that there are fundamental algorithmic limitations that prevent such schemes, which may very well be the case in light of recent trends in theoretical computer science \cite{barak2016nearly}. In either case, in practice there are often many problem-specific structural assumptions to be made about a signal, and thus initialization for sparse signals may be obtained via other means that may not require $O(k^2 \log n)$ magnitude measurements. 

Lastly, the proof of our main theorem illustrates an intriguing connection between phaseless compressed sensing and 1-bit compressed sensing. In fact we rely on a crucial theorem in \cite{yaniv} to analyze the optimality conditions of SparsePhaseMax. A duality between phaseless and phase-only measurements has been observed before in a different setting, by the authors of \cite{bandeira2014saving}, in their analysis of injectivity of phase retrieval. Specifically, they established that measurement vectors lead to injective measurements modulo phase if and only if they lead to injective measurements modulo magnitude. 

\section{Proof}

In the course of proving our main theorem, we establish a connection between phaseless and 1-bit compressed sensing. Specifically, the optimality conditions of SparsePhaseMax involve the signs of the quantities $\langle a_i,x_0 \rangle$ and $\langle a_i,h \rangle$ where $h$ is a perturbation about $\xo$, and we use a key theorem from \cite{yaniv} to control the set of admissible ascent directions. This theorem states that if two approximately $s$-sparse vectors in $\R^n$ are on the same side of $O(s \log(2n/s))$ random hyperplanes, then these two vectors are approximately equal with high probability.  We state it here in full for convenience:


\begin{theorem}[Theorem 2.1 in \cite{yaniv}] \label{thm:yaniv}
Let $n,m,s > 0$ and set $\delta = \Ctilde_1\left( \frac{s}{m} \log(2n/s) \right)^{1/5}.$  Let $a_i \sim \mathcal{N}(0, I_{n \times n})$ be independent for $i = 1 \ldots m$.  Then with probability at least $1 - \Ctilde_2 \exp(-\ctilde \delta m)$, the following
holds uniformly for all $x, \xhat \in \R^n$ that satisfy $\|x\|_2 = \|\xhat\|_2 = 1$, $\|x\|_1 \leq \sqrt{s}$, and $\|\xhat\|_1 \leq \sqrt{s}$, for $s \leq n$ :
\[
\langle a_i, \xhat \rangle \langle a_i, x \rangle \geq 0 ,  \ \  i=1\ldots m  \quad \implies \quad  \| \xhat - x \|_2 \leq \delta.
\]
Here, $\Ctilde_1, \Ctilde_2, \ctilde$ are positive universal constants. 
\end{theorem}

We now proceed to establish our main result:

\begin{proof}[Proof of Theorem \ref{thm:sparsephasemax-smallalpha}]
Without loss of generality, let $\|\xo\|_2 = 1$, and let $\Omega$ be the support of $x_0$.   To establish that $x_0$ is the unique global maximizer,  consider a feasible $\xo+h$ such that $\langle \phi, \xo + h \rangle - \lambda \| \xo + h \|_1 \geq \langle \phi, \xo\rangle - \lambda \| \xo \|_1$.  We will show that $h = 0$.  Write $h = \homega + \homegac$, where these terms have support in $\Omega$ and $\Omega^c$, respectively. Because $\xo +h$ is feasible, we have
\begin{align}
\langle a_i, \xo \rangle \langle a_i, h \rangle \leq 0, \ \ i =1 \ldots m. \label{optimality-conditions}
\end{align}

Because $\sign(\xoomega) + \sign(\homegac)$ is a subgradient of $\|\cdot\|_1$ at $\xo$, we have
\begin{align*}
\langle \phi, h \rangle &- \lambda \|x_0\|_1 - \lambda \langle \sign(\xoomega) + \sign(\homegac), h \rangle \geq - \lambda \| x_0\|_1.
\intertext{Thus, }
\langle \phi, h\rangle  &\geq \lambda \langle \sign(\xoomega), \homega \rangle + \lambda \langle \sign \homegac,  \homegac \rangle\\
 & = \lambda \langle \sign(\xoomega), \homega \rangle + \lambda \| \homegac\|_1.
\end{align*}
By writing $h = \langle \xo, h \rangle \xo + \hxop$, where $\hxop$ is the projection of $h$ onto the orthogonal complement of span$(\{\xo\})$, we have
\begin{align}
\|\homegac\|_1 &\leq  - \langle \sign(\xoomega), \homega \rangle + \frac{1}{\lambda} \langle \phi, \xo \rangle \langle \xo, h \rangle + \frac{1}{\lambda} \langle \phi, \hxop \rangle \notag\\
&= - \langle \sign(\xoomega), \homega \rangle + \frac{\alpha}{\lambda} \langle \xo, \homega\rangle +  \frac{1}{\lambda} \langle \phi, \hxop \rangle \notag\\
&\leq \sqrt{k} \| \homega\|_2 + \frac{\alpha}{\lambda} \langle \xo, \homega\rangle + \frac{1}{\lambda} \| \hxop\|_2 \label{lowerbound-homegac} \\
&\leq \sqrt{k} \| \homega\|_2 + \frac{\alpha}{\lambda} \|\homega\|_2 +  \frac{1}{\lambda} \| \hxop\|_2 \notag \\
&\leq \Bigl(1 + \frac{\alpha}{\xi} + \frac{1}{\xi} \Bigr) \sqrt{k} \|h\|_2 \notag \\
&\leq \Bigl(\frac{\xi}{\xi} + \frac{1}{\xi}  + \frac{1}{\xi}\Bigr) \sqrt{k} \|h\|_2 \notag\\
&\leq \frac{11/4}{\xi} \sqrt{k} \| h\|_2, \label{lowerbound-homegac-preyaniv}
\end{align}
where the second line follows by definition of $\alpha$ and because $\xo$ is supported on $\Omega$; the third line follows by Cauchy-Schwarz, the assumption that $|\Omega|=k$, and the assumption that $\|\phi\|_2=1$; the fourth line follows because $\|\xo\|_2=1$, the fifth line follows by definition of $\xi := \lambda \sqrt{k}$, the sixth line follows by $\alpha \leq 1$, and the last line follows by the assumption that $\xi < \frac{3}{4} \alpha \leq \frac{3}{4}$. 
We now compute
\begin{align*}
\|h\|_1 &= \| \homegac\|_1 + \|\homega\|_1\\
&\leq \frac{11/4}{\xi} \sqrt{k} \| h\|_2 + \sqrt{k} \| \homega\|_2 \\
&\leq \Bigl( \frac{11/4}{\xi} + 1 \Bigr) \sqrt{k} \|h\|_2\\
&\leq  \frac{3.5}{\xi} \sqrt{k} \|h\|_2,
\end{align*}
where the second line follows by \eqref{lowerbound-homegac-preyaniv} and $|\Omega|=k$,  and the last line follows by $\xi < \frac{3}{4} \alpha \leq \frac{3}{4}$.
Assuming $h \neq 0$, 
we arrive at 
\[
\frac{\|h\|_1}{\|h\|_2} \leq \frac{3.5}{\xi} \sqrt{k}.
\]

Now, if $m > \frac{3.5^2\cdot \Ctilde_1^5}{(\alpha/20)^5 \xi^2} k \log \frac{2 n \xi^2}{3.5^2 k}$ then  $\delta := \Ctilde_1(\frac{s}{m} \log \frac{2n}{s})^{1/5}<\frac{\alpha}{20}$, where $s = \frac{3.5^2}{\xi^2} k$.  Further, as $k <  \frac{1}{49} \alpha^2 n $ and $\xi > \frac{\alpha}{2}$, we have that $s < n$ and that there exists a universal constant $C_1$ such that $m > \frac{C_1}{\alpha^7} k \log \frac{n}{k}$  implies  $m > \frac{3.5^2\cdot \Ctilde_1^5}{(\alpha/20)^5 \xi^2} k \log \frac{2 n \xi^2}{3.5^2 k}$ and thus $\delta < \frac{\alpha}{20}$. 
By Theorem~\ref{thm:yaniv},  the feasibility conditions \eqref{optimality-conditions} imply that there exists an event $E$ with probability 
\begin{align}
\mathbb{P}(E) \geq 1 - \Ctilde_2 e^{-\ctilde \delta m} &= 
1 - \Ctilde_2 \exp \left[-\frac{3.5^{2/5} \ctilde \Ctilde_1}{\xi^{2/5}} m^{4/5} k^{1/5} \log^{1/5} \Bigl(\frac{8 n \xi^2}{49 k} \Bigr) \right] \notag\\
&\geq 
1 - \Ctilde_2 \exp \left[-\frac{3.5^{2/5} \ctilde \Ctilde_1}{(3/4)^{2/5}} \alpha^{-2/5} m^{4/5} k^{1/5} \log^{1/5} \Bigl(\frac{ 2 n\alpha^2}{ 49 k} \Bigr) \right] \notag\\
&\geq 
1 - \Ctilde_2 \exp \left[-\frac{3.5^{2/5} \ctilde \Ctilde_1}{(3/4)^{2/5}} \alpha^{-2/5} m^{4/5} k^{1/5} (\log 2)^{1/5} \right], \notag
\end{align}
on which $\left\| \xo + \frac{h}{\|h\|_2} \right\|_2 \leq \delta$.  To get the inequalities above, we use $\xi < \frac{3}{4} \alpha$, $\xi > \frac{1}{2} \alpha$, and the assumption that $\frac{n \alpha^2}{k} > 49$.  Hence, on $E$,
\begin{align}
\langle\xo,  h \rangle &\leq -(1-\delta) \|h\|_2  \text{ and  }\| \hxop\|_2 \leq \delta \|h\|_2. \label{boundsonE}
\intertext{By \eqref{lowerbound-homegac}, }
\| \homegac\|_1 &\leq  \sqrt{k} \| \homega\|_2 + \frac{\alpha}{\lambda} \langle \xo, \homega\rangle + \frac{1}{\lambda} \| \hxop \|_2 \notag\\
&\leq \sqrt{k} \| \homega\|_2 - \frac{\alpha}{\lambda} (1-\delta) \| h \|_2 + \frac{\delta}{\lambda} \| h \|_2 \notag\\
&\leq \Bigl( 1 - \frac{\alpha}{\xi} (1-\delta) + \frac{\delta}{\xi} \Bigr) \sqrt{k} \| h\|_2 \notag \\
&\leq \Bigl( 1 - \frac{4}{3} (1-\delta) + \frac{\alpha}{20 \xi}  \Bigr) \sqrt{k} \|h\|_2 \notag \\
&\leq \Bigl( -\frac{1}{3} + \frac{4}{3}\cdot \frac{1}{20} + \frac{1}{10}  \Bigr)  \sqrt{k} \| h\|_2 \notag \\
&\leq -0.16 \sqrt{k} \|h\|_2, \notag
\end{align}
where the second line follows from \eqref{boundsonE}, the third line follows from the definition of $\xi$, the fourth line follows from $\xi < \frac{3}{4} \alpha$ and $\delta < \frac{\alpha}{20}$, the fifth line follows from $\delta < \frac{1}{20}$ and $\xi > \frac{1}{2} \alpha$.   Hence, $\|\homegac\|_1=0$, from which we also conclude $\|h\|_2 = 0$.

\end{proof}

Corollary \ref{cor:sparsephasemax-smallalpha} follows immediately from Theorem \ref{thm:sparsephasemax-smallalpha} by the probability bound $1 - C_2 e^{-c \alpha^{-2/5} m^{4/5} k^{1/5} }  \geq 1 - C_2 e^{-c k}$, where we use  $\alpha \leq 1$ and $m \geq k$.

\subsection*{Acknowledgements}
PH and VV  thank Thang Huynh for discussions and improvements to the proof.  PH acknowledges funding by the grant NSF DMS-1464525. 
\bibliographystyle{plain}
\bibliography{refs}

\end{document}